\documentclass{article}
\usepackage[utf8]{inputenc}
\usepackage{authblk}

\title{Involutions of Halphen Pencils of Index 2 and Discrete Integrable Systems}
\author{Kangning Wei }
\affil{Department of Mathematics, Technische Universit\"at Berlin,\\ Str. des 17. Juni, Berlin, Germany}
\date{}
\usepackage{amssymb}  
\usepackage{amsthm}
\usepackage{amsmath}
\usepackage{hyperref}

\newtheorem{prop}{Proposition}
\newtheorem{thm}{Theorem}
\newtheorem{claim}{Claim}
\newtheorem{rmk}{Remark}
\begin{document}

\maketitle
\begin{abstract}
We constructed involutions for a Halphen pencil of index 2, and proved that the birational mapping corresponding to the autonomous reduction of the elliptic Painlev\'e equation for the same pencil can be obtained as the composition of two such involutions.
\end{abstract}

\section{Introduction}
The discrete elliptic Painlev\'e equation, defined in \cite{Sakai2001}, as a discrete dynamical system, can be described by a birational mapping (depending on parameters), which we shall refer to as the elliptic Painlev\'e mapping. This mapping is often referred to as the non-autonomous version of the QRT mapping. The QRT mapping, often defined on $\mathbb{P}^1\times\mathbb{P}^1$, preserves a pencil of biquadratic curves, which corresponds to a pencil of cubic curves on $\mathbb{P}^2$ under the canonical birational equivalence of $\mathbb{P}^1\times\mathbb{P}^1$ and $\mathbb{P}^2$. Similarly, the elliptic Painlev\'e mapping also preserves a pencil of cubic curves at each step of the time evolution, but this pencil of cubic curves changes from one step to another. It is in this sense that we say the elliptic Painlev\'e mapping is the non-autonomous version of the QRT mapping.\\

An elliptic Painlev\'e mapping comes with nine parameter points, whose configuration determines its dynamical properties \cite{Kajiwara_2017}\cite{Sakai2001}. For some special configurations of the parameter points, one can find autonomous reductions of the elliptic Painlev\'e mapping. More precisely, if we start with nine parameter points which are the base points of a index $k$ Halphen pencil, i.e. they are the base points of a pencil of degree $3k$ curves, each having multiplicity $k$, then the $k$-th iteration of the elliptic Painlev\'e mapping becomes autonomous and preserves the Halphen pencil \cite{Kajiwara2006}.\\

In the QRT case, the QRT mapping is written as the composition of a horizontal and a vertical switch. For a pencil of cubic curves, one can define Manin involutions in terms of the group law on cubic curves. It is shown that the QRT mapping can also be realized as the composition of two generalized Manin involutions \cite{Kamp2018GeneralisedMT}. In \cite{Petrera2021}, the geometric construction of Manin involutions is generalized to certain pencils of elliptic curves of degree 4 and degree 6, which are birationally equivalent to a pencil of cubics and comes naturally from Kahan discretizations.\\

By definition, pencils of cubic curves are Halphen pencils of index 1. In this paper, we construct involutions on Halphen pencils of index 2, which are pencils of sextic curves with nine double base points. These involutions are known as Bertini involutions in the literature \cite[p.127]{Hudson}.
Our main result is that the autonomous reduction of the elliptic Painlev\'e mapping for a index 2 Halphen pencil can be realized as the composition of two Bertini involutions.

\section{Elliptic Painlev\'e Mapping}
There are several different ways to define the elliptic Painlev\'e mapping in the literature. Here we follow the geometric description from \cite{Kajiwara2006}, without mentioning the rather involved construction using root systems on the Picard lattice of the blowing up surface \cite{Kajiwara_2017}\cite{Sakai2001}. The statements in this section can also be proved using the algebraic geometry tools \cite{Kajiwara_2017}. To make this paper self-contained, we give proofs entirely based on the geometric descriptions.\\     

Let $P_1,...,P_9$ be nine points on the projective plane $\mathbb{CP}^2$ in generic position, i.e. there is a unique cubic curve $C_0$ passing through the nine points $P_1,...,P_9$. They are considered as parameters, and $P_{10}\in\mathbb{CP}^2$ is the dependent variable. Then the elliptic Painlev\'e mapping $$T_{ij}:\{P_1,...,P_{10}\}\mapsto\{\bar{P_1},\bar{P_2},...,\bar{P_{10}}\}$$
has the following geometric description:
\begin{itemize}
    \item The transformation of the parameters $P_1,...,P_9$ under $T_{ij}$ is determined by
    \begin{equation}\label{eq:1}
        \bar{P_k}=P_k\; for \; k\neq i,j
    \end{equation}
    \begin{equation}\label{eq:2}
        P_1+..+P_{j-1}+\bar{P_j}+P_{j+1}+...+P_9=0
    \end{equation}
    \begin{equation}\label{eq:3}
 \bar{P_i}+\bar{P_j}=P_i+P_j
 \end{equation}

    where the addition is taken as the group law on the cubic curve $C_0$
    \begin{rmk}
 Equation \eqref{eq:2} means the nine points $P_1,..,P_{j-1},\bar{P_j},P_{j+1},..,P_9$ are the base points of a pencil of cubic curves.
    \end{rmk}
    \item The dependent variable $P_{10}$ determines a cubic curve $C_{P_{10}}$ passing through the nine points $P_1,..,\hat{P_j},..,P_{10}$, where $\hat{P_j}$ means $P_j$ is deleted, and under $T_{ij}$, $P_{10}$ transforms according to
    \begin{equation}\label{eq:4}
            \bar{P_{10}}+\bar{P_j}=P_{10}+P_i
    \end{equation}
    where the addition law is taken as the group law on the cubic curve $C_{P_{10}}$
    \begin{rmk}
The transformation of the parameters points $P_1,...,P_9$ does not depend on the transformation of the variable $P_{10}$.
\end{rmk}
\end{itemize}

\begin{rmk}
The definition does not depend on the choice of the zero element for the group law on the cubic curve. Since the equation \eqref{eq:2} does not depend on the choice of zero, and the equations \eqref{eq:3} and \eqref{eq:4} are both translations on a cubic curve, which also does not depend on the choice of zero.
\end{rmk}

Although we assumed that the nine parameter points $P_1,..,P_9$ are in general position, the map $T_{ij}$ is still well-defined even if $P_1,..,P_9$ are base points of a pencil of cubic curves. In fact, in that case, the equations \eqref{eq:1} \eqref{eq:2} and \eqref{eq:3} imply that $P_1,..,P_9$ are fixed by $T_{ij}$. So the map becomes autonomous.\\

More generally, we have the following statement:
\begin{prop}
If the nine parameter points $P_1,...,P_9$ are the base points of a Halphen pencil of index $k$, i.e., in terms of the addition law on $C_0$, we have $k(P_1+P_2+...P_9)=0$, then the $k$-th iteration of the map $T_{ij}$ is autonomous:
$$T^{k}_{ij}(P_l)=P_l\; for\; any\; l\in\{1,2,..,9\}$$
\end{prop}
\begin{proof}
We only need to prove for $P_i$ and $P_j$.
Denote by $\delta=P_1+P_2+...+P_9$, then the equation \eqref{eq:2} can be written as $\bar{P_j}-P_j=-\delta$, plugging into the equation
\eqref{eq:3}, we obtain
$$\bar{P_i}=P_i+P_j-\bar{P_j}=P_i+\delta$$
Since $P_i+P_j$ is invariant, so is $\delta$, thus we have
$$T_{ij}^k(P_i)=P_i+k\delta=P_i$$
The assertion for $P_j$ holds since $P_i+P_j$ is invariant.
\end{proof}

\section{Main Result}
Now we focus on Halphen pencils of index 2, which are described by nine points $P_1,..,P_9$ on $\mathbb{CP}^2$ satisfying $2(P_1+P_2+...+P_9)=0$, where the addition is taken on the unique cubic curve passing through the nine points. There is a one parameter family of sextic curves having double points at $P_1,...,P_9$, which is the Halphen pencil of index 2 determined by the base points. 
The unique cubic curve $C_0$ passing through the base points $P_1,...,P_9$ is also contained in this pencil as a double curve.We denote this pencil by $\mathcal{H}=\mathcal{P}(6;P_1^2,...,P_9^2)$.\\

On a Halphen pencil of index 2, we can define the following involutions, which we denote by $I_i$, labeled by one of the base points $P_i$. 
\begin{itemize}
    \item For a generic point $P\in\mathbb{P}^2$, there is a unique sextic curve $H_P\in\mathcal{H}$ passing through $P$, and a unique cubic curve passing through $P_1,...,\hat{P_i},..,P_9,P$, denoted by $C^i_P$.
    \item The curve $H_P$ and $C^i_P$ has 18 intersections, counted with multiplicity, 17 of them are given by the base points $P_1,...\hat{P_i},...,P_9$ with multiplicity 2, and $P$.
    \item We define $I_i(P)$ to be the unique other intersection of $C^i_P$ and $H_P$.
\end{itemize}
\begin{rmk}
The involutions $I_{i}$s are known as Bertini involutions in the literature. We refer to \cite{Bayle2000} and \cite{Hudson} for details. 
\end{rmk}

The following is our main result:
\begin{thm}\label{MainThm}
Let $P_1,...,P_9$ be the base points of a Halphen pencil of index 2. Then the autonomous map $T_{ij}^2:\mathbb{P}^2\rightarrow\mathbb{P}^2$, seen as a map of $P_{10}$, is equal to the composition of involutions
$$T_{ij}^2=I_j\circ I_i$$
\end{thm}
Since $T_{ij}$ is not autonomous, in the expression $T_{ij}^2$, when considered as a map of $P_{10}$, the first and second step are actually different maps, we will denote the first step as $T_{ij}^{(1)}$, and the second step as $T_{ij}^{(2)}$, then the statement is written as $T_{ij}^{(2)}\circ T_{ij}^{(1)}=I_j\circ I_i$.\\

In order to prove the statement, we introduce another involution, defined via the intersection of the following two pencils:
\begin{itemize}
    \item The pencil of cubic curves passing through $P_1,..,\hat{P_i},..,P_9$, denoted by $\mathcal{C}^i$
    \item The pencil of cubic curves passing through $P_1,..,\hat{P_j},..,P_9$, denoted by $\mathcal{C}^j$
\end{itemize}

Generic curves of the two pencils will have 9 intersection points, with 7 of them being the base points $P_1,..,\hat{P_i},..,\hat{P_j},..,P_9$, so we define the involution $J_{ij}$ by the following:
\begin{itemize}
    \item For a generic point $P\in\mathbb{P}^2$, there is a unique curve of the pencil $\mathcal{C}^i$ passing through $P$, denoted by $C^i_P$, and a unique curve of the pencil $\mathcal{C}^j$ passing through $P$,denoted by $C^j_P$
    \item The curve $C^i_P$ and $C^j_P$ have 9 intersections, counted with multiplicities, with 8 of them given by the 7 base points $P_1,..,\hat{P_i},..,\hat{P_j},..,P_9$ and $P$
    \item We define $J_{ij}(P)$ to be the unique other intersection of the two curves $C^i_P$ and $C^j_P$
\end{itemize}
Notice that both of the involution $J_{ij}$ and the Halphen involution $I_j$ are involutions preserving the pencil $\mathcal{C}^j$, while the first step elliptic Panilev\'e map $T_{ij}^{(1)}$ is a translation on the pencil $\mathcal{C}^j$. We claim the following:
\begin{prop}
The map $T_{ij}^{(1)}$ (as a map of $P_{10}$) is the composition of the two involutions $J_{ij}$ and $I_j$:
$$T_{ij}^{(1)}=J_{ij}\circ I_j$$
\end{prop}
\begin{proof}
The elliptic Painlev\'e mapping $T_{ij}^{(1)}$, as a birational map of $P=P_{10}$ on $\mathbb{P}^2$, is given by 
$$\bar{P}=T_{ij}^{(1)}(P)=P+P_i-\bar{P_j}$$
where the group law is taken on the cubic curve $C^j_P$ of the pencil $\mathcal{C}^j$. 
In other words,for any given $P$, $T_{ij}^{(1)}$ is a translation on the curve $\mathcal{C}^j_P$ by $P_i-\bar{P_j}$.\\
To prove the statement, it is then sufficient to show that the composition $J_{ij}\circ I_j$ is also the same translation. Since both $J_{ij}$ and $I_j$ are involutions preserving the pencil $\mathcal{C}^j$, the composition is automatically a translation on the pencil. It suffices to prove that $J_{ij}\circ I_j(\bar{P_j})=P_i$.\\
So the proposition is proved if we can prove the following two statements
\begin{claim}
$I_j(\bar{P_j})$ is well defined and
$I_j(\bar{P_j})=\bar{P_j}$
\end{claim}
\begin{claim}
$J_{ij}(\bar{P_j})$ is well defined and
$J_{ij}(\bar{P_j})=P_i$
\end{claim}
To prove the first claim, we observe that $\bar{P_j}$ is a base point of the pencil $\mathcal{C}^j$, but it is not a base point of the Halpehn pencil $\mathcal{H}$. It follows that there is a unique curve $H_{\bar{P_j}}$ of the Halphen pencil $\mathcal{H}$ passing through $\bar{P_j}$, which is exactly the double curve of the cubic curve $C_0$ which passes through $P_1,..,P_9$ (also passing through $\bar{P_i}$ and $\bar{P_j}$ by definition of the elliptic Painlev\'e mapping). So for any generic curve of the pencil $\mathcal{C}^j$, its intersections with $H_{\bar{P_j}}$ are given by the nine points $P_1,..,\bar{P_j},..,P_9$, each with multiplicity 2. It follows from definition that $I_j(\bar{P_j})=\bar{P_j}$.\\
To prove the second claim, we observe that $\bar{P_j}$ is not a base point of the pencil $\mathcal{C}^i$, there is one unique curve $C^i_{\bar{P_j}}$ of the pencil $\mathcal{C}^i$ passing through $\bar{P_j}$, which is exactly the cubic curve $C_0$ which passes through $P_1,..,P_9$ and $\bar{P_i},\bar{P_j}$.So for any generic curve of the pencil $\mathcal{C}^j$, its intersections with $C^i_{\bar{P_j}}$ are the nine points $P_1,..,\bar{P_j},...,P_9$.It follows from the definition that $J_{ij}(\bar{P_j})=P_i$.
\end{proof}
Similarly, we have the following proposition
\begin{prop}
The map $T_{ij}^{(2)}$ (as a map of $P_{10}$) is the composition of the two involutions $I_i$ and $J_{ij}$:
$$T_{ij}^{(2)}=I_i\circ J_{ij}$$
\end{prop}
\begin{proof}
The map $T_{ij}^{(2)}$,as a birational map on $\mathbb{P}^2$, is defined by
$$T_{ij}^{(2)}(P)=P+\bar{P_i}-P_j$$
where the group law is taken on the cubic curve $C^i_P$ of the pencil $\mathcal{C}^i$. In other words,for any given $P$,$T_{ij}^{(2)}$ is a translation on the curve $\mathcal{C}^i_P$ by $\bar{P_i}-P_j$.\\
Similarly, the proposition follows from the following two statements:
\begin{claim}
$J_{ij}(P_j)$ is well defined and $J_{ij}(P_j)=\bar{P_i}$
\end{claim}
\begin{claim}
$I_i(\bar{P_i})$ is well defined and $I_i(\bar{P_i})=\bar{P_i}$
\end{claim}
The first claim is proved in exactly the same fashion as before, by observing that $P_j$ is not a base point of the pencil $\mathcal{C}^j$ and the unique cubic curve of the pencil $\mathcal{C}^j$ passing through $P_j$ is the cubic $C_0$ passing through $P_1,...,P_9$ and $\bar{P}_i,\bar{P}_j$. So for any generic curve of the pencil $\mathcal{C}^i$, its intersections with $C_0$ are the nine points $P_1,..,\bar{P_i},..,P_j,..,P_9$, and by definition $J_{ij}(P_j)=\bar{P_i}$.\\
The second claim also follows similarly from the observation that $\bar{P_i}$ is not a base point of the Halphen pencil $\mathcal{H}$ and the unique Halphen curve $H_{\bar{P_i}}$ of the pencil $\mathcal{H}$ passing through $\bar{P_i}$ is the double curve of the cubic curve $C_0$. So any generic curve of the pencil $\mathcal{C}^i$ intersects with the curve $H_{\bar{P_i}}$ at the nine points $P_1,..,\bar{P_i},..,P_9$,each with multiplicity 2.It follows from the definition that $I_i(\bar{P_i})=\bar{P_i}$. 
\end{proof}
Combining the propositions, the proof of the theorem now follows easily as $$T_{ij}^{(2)}\circ T_{ij}^{(1)}=I_i\circ J_{ij}\circ J_{ij}\circ I_j=I_i\circ I_j$$
since $J_{ij}$  is a involution. \\
\begin{rmk}
We remark here that it does not seem to be possible to define similar involutions for Halphen pencil of index higher than 2. The involutions we defined make use of another pencil of curves whose intersection with the Halphen pencil generically having exactly 2 points besides the base points. Or in more precise terms, the divisor classes  of these two pencils have intersection number 2. However, for Halphen pencil with higher index, such pencil does not exist as the divisor class of the Halphen pencil will be of the form $-nK$ with $n\geq3$ and $K$ being the canocial disvisor.
\end{rmk}

\section{Example: HKY Mapping}
This example is taken from \cite{article}, known as an HKY mapping. In \cite{Carstea_2012}, a symmetric version of this map is considered, and is classified in the (i-2) class, i.e., mappings which preserve a index 2 Halphen pencil.\\

The map is defined by the recurrence relation:
\begin{equation}
    y_ny_{n-1}=\frac{x_n^2-t^2}{sx_n-1}
\end{equation}
\begin{equation}
        x_{n+1}x_n=\frac{y_n^2-t^2}{y_n/s-1}
\end{equation}
The map $\Phi_1:(x_n,y_n)\mapsto(x_{n+1},y_n)$ and $\Phi_2:(x_{n+1},y_n)\mapsto(x_{n+1},y_{n+1})$ are both involutions.
\begin{rmk}
In the symmetric case,i.e.,s=1, we denote $X_{2n}=x_n$, $X_{2n-1}=y_n$, then the map $\Phi:(X_k,X_{k+1})\mapsto(X_{k+1},X_{k+2})$ is the composition of $\Phi_1$ or $\Phi_2$ with a symmetry switch. In comparison with the QRT case, the two involutions $\Phi_1$ and $\Phi_2$ corresponds to the horizontal and vertical switches, while in the symmetric case the map $\Phi$ corresponds to the QRT root. 
\end{rmk}
Rewrite the two maps in homogeneous coordinates:
\begin{equation}
    \Phi_1:[x:y:z]\mapsto [s(y^2-t^2z^2)z:xy(y-sz):xz(y-sz)]
\end{equation}
\begin{equation}
    \Phi_2:[x:y:z]\mapsto [xy(sx-z):(x^2-t^2z^2)z:yz(sx-z)]
\end{equation}
The HKY mapping is the composition $\Phi_2\circ\Phi_1$.\\

This map preserves an index 2 Halphen pencil. To see this, one can performs nine blowing ups at the following points to lift the map to a surface automorphism.  
\begin{equation}
\begin{split}
     P_1=[0:-t:1] \\
     P_2=[0:t:1] \\
     P_3=[-t:0:1] \\
     P_4=[t:0:1] \\
     P_5=[1:0:0] \\
     P_6=[0:1:0] \\
     P_7=[s^2:1:0]\\
\end{split}
\end{equation}
\begin{itemize}
    \item The point $P_8$ is infinitely close to $P_5$, given by the direction $\{y=sz\}$
    \item The point $P_9$ is infinitely close to $P_6$, given by the direction $\{z=sx\}$
\end{itemize}
Denoting the corresponding divisors coming from the blowing up at $P_i$ by $E_i$.
Then the singularity confinement pattern of $\Phi_1$ and $\Phi_2$ is summarized as follows:
\begin{equation}
    \begin{split}
        \overline{\{y=sz\}} \xrightarrow{\text{$\Phi_1$}} E_8 \xrightarrow{\text{$\Phi_2$}} E_7 \xrightarrow{\text{$\Phi_1$}} E_9 \xrightarrow{\text{$\Phi_2$}} \overline{\{z=sx\}}  \\
        \overline{\{x=0\}} \xrightarrow{\text{$\Phi_1$}} E_5-E_8 \xrightarrow{\text{$\Phi_2$}} \overline{\{z=0\}} \xrightarrow{\text{$\Phi_1$}} E_6-E_9 \xrightarrow{\text{$\Phi_2$}} \overline{\{y=0\}} \\
     \overline{\{y=tz\}} \xrightarrow{\text{$\Phi_1$}} E_2 \xrightarrow{\text{$\Phi_2$}} E_2 \xrightarrow{\text{$\Phi_1$}} \overline{\{y=tz\}} \\
        \overline{\{y=-tz\}} \xrightarrow{\text{$\Phi_1$}} E_1 \xrightarrow{\text{$\Phi_2$}} E_1 \\
         \overline{\{x=tz\}} \xrightarrow{\text{$\Phi_2$}} E_4 \xrightarrow{\text{$\Phi_1$}} E_4 \\
        \overline{\{x=-tz\}} \xrightarrow{\text{$\Phi_2$}} E_3 \xrightarrow{\text{$\Phi_1$}} E_3 \\
    \end{split}
\end{equation}
Then the mappings $\Phi_1$ and $\Phi_2$ are exactly the involutions $I_4=I_3$ and $I_1=I_2$, which we defined for an index 2 Halphen pencil.

\section{Conclusion}
In this paper, we defined a certain type of involutions on a index 2 Halphen pencil. And we showed that the compositions of these involutions are twice iterations of elliptic Painlev\'e mapping whose nine parameter points being the base points of the same Halphen pencil. In this sense, the involutions we defined for index 2 Halphen pencils play the same role as the Manin involtuions defined on a cubic pencil, where compositions of the involutions give rise to the QRT mapping. \\

\section{Acknowledgement}
This research is supported by the DFG Collaborative Research Center TRR 109 “Discretization in Geometry and Dynamics”.
\bibliographystyle{plain}
\bibliography{Bibtex.bib}

\begin{thebibliography}{1}

\bibitem{Bayle2000}
Lionel Bayle and Arnaud Beauville.
\newblock Birational involutions of $\mathbf{P}^2$.
\newblock {\em Asian Journal of Mathematics}, 4(1):11--18, 2000.

\bibitem{Carstea_2012}
A~S Carstea and T~Takenawa.
\newblock A classification of two-dimensional integrable mappings and rational
  elliptic surfaces.
\newblock {\em Journal of Physics A: Mathematical and Theoretical},
  45(15):155206, mar 2012.

\bibitem{Hudson}
Hilda~P. {Hudson}.
\newblock {Cremona transformations in plane and space}.
\newblock {XX + 454 p. Cambridge, University Press (1927).}, 1927.

\bibitem{Kajiwara2006}
K.~Kajiwara, Tetsu Masuda, M.~Noumi, Y.~Ohta, and Y.~Yamada.
\newblock Point configurations, cremona transformations and the elliptic
  difference painleve equation.
\newblock {\em Seminaires et Congres}, 14:169--198, 2006.

\bibitem{Kajiwara_2017}
Kenji Kajiwara, Masatoshi Noumi, and Yasuhiko Yamada.
\newblock Geometric aspects of painlev{\'{e}} equations.
\newblock {\em Journal of Physics A: Mathematical and Theoretical},
  50(7):073001, jan 2017.

\bibitem{article}
K~Kimura, H~Yahagi, R~Hirota, A~Ramani, B~Grammaticos, and Y~Ohta.
\newblock A new class of integrable discrete systems.
\newblock {\em Journal of Physics A: Mathematical and General}, 35:9205, 10
  2002.

\bibitem{Petrera2021}
Matteo Petrera, Yuri~B. Suris, Kangning Wei, and Ren{\'e} Zander.
\newblock Manin involutions for elliptic pencils and discrete integrable
  systems.
\newblock {\em Mathematical Physics, Analysis and Geometry}, 24(1):6, Mar 2021.

\bibitem{Sakai2001}
Hidetaka Sakai.
\newblock Rational surfaces associated with affine root systems and geometry of
  the painlev{\'e} equations.
\newblock {\em Communications in Mathematical Physics}, 220(1):165--229, Jun
  2001.

\bibitem{Kamp2018GeneralisedMT}
Peter~H. van~der Kamp, D.~I. McLaren, and G.~Quispel.
\newblock {Generalised Manin transformations and QRT maps}.
\newblock {\em arXiv: Exactly Solvable and Integrable Systems}, 2018.

\end{thebibliography}

\end{document}